\documentclass[notitlepage]{revtex4-1}
\usepackage{amssymb}
\usepackage{graphicx}
\usepackage{enumerate}
\usepackage{amsmath}
\usepackage{amsthm}
\usepackage[hidelinks=true]{hyperref} 
\hypersetup{
    colorlinks,
    citecolor=red,
    filecolor=black,
    linkcolor=blue,
    urlcolor=black
}
\newtheorem{theorem}{Theorem}[section]

\begin{document}
\title{A Brief Review of Helmholtz Conditions}
\author{Kushagra Nigam}
\email{f2010174@goa.bits-pilani.ac.in}
\author{Kinjal Banerjee}
\email{kinjalb@gmail.com}
\affiliation{Department of Physics, BITS Pilani K.K Birla Goa Campus, N.H. 17B Zuarinagar, Goa 403726, India.}
\affiliation{Department of Physics, BITS Pilani K.K Birla Goa Campus, N.H. 17B Zuarinagar, Goa 403726, India.}
\date{\today}

\begin{abstract}
It is well known that the equations of motion obtained from Newtons second law of motion can be obtained 
from a Lagrangian via the Euler-Lagrangian formulation if and only if the equations of motion satisfy the Helmholtz 
conditions. In this pedagogical article we give a simple proof of the above statement and show its application to simple 
mechanical and dissipative systems. 
\end{abstract}
\maketitle
\vspace{-5ex}
\tolerance=1 
\emergencystretch=\maxdimen
\hyphenpenalty=10000
\hbadness=10000
\section{Introduction}
The main aim of classical mechanics is to be able to predict with certainty, the final state of a system, i.e. to
predict its position and velocity at some time $t$, if we are given its position and  velocity at some initial time
$t_0$ and the set of forces acting on the system. In other words, the goal is to determine the equations of motion and
then obtain the solutions. However the starting point of quantum theory is usually the Hamiltonian or Lagrangian. 
So, given the equations of motion of a classical system, it is first necessary to {\textit{invert}} them to obtain a 
Lagrangian or a Hamiltonian before proceeding to quantization. A natural question which arises here is that: Is it always possible 
to find a Lagrangian corresponding to a given set of equations of motion? More specifically, we need to know whether a given set of 
second order ordinary differential equations governing the dynamics of a physical system can be obtained as Euler-Lagrange equations 
of some Lagrangian function.

The necessary and sufficient conditions for this to be so are known as Helmholtz conditions. These conditions have been extensively 
studied in mathematical literature by \cite{davis1,davis2}, \cite{jessedouglas} among others. Recently, various attempts have been made to 
find integrating factors which when multiplied to equation of motion lead to the existence of a Lagrangian
\cite{leonardo}. A brief discussion of generalized Helmholtz conditions can also be found in \cite{crampin}.

A standard undergraduate course on Classical Mechanics introduces the concept of Lagrangian and Action integral. 
However, very often, the question of the existence of a Lagrangian function is completely ignored. 
Keeping these things in mind, our goal is to provide a pedagogical account of the problem. 
The proof, though slightly involved and lengthy, uses simple multivariate calculus which can be easily understood by undergraduate students. 
The conditions themselves put stringent restrictions on what type of mechanical systems can be studied using 
the action formalism. One of the places where we can come across non-Lagrangian systems is while studying phenomena like
dissipation. In this paper, we show the application of the Helmholtz conditions through an example of determining the
Lagrangian of a damped harmonic oscillator.  

This paper is organized as follows: 
In Section (\ref{proof}) we give a proof of necessity and sufficiency of Helmholtz conditions in $n$ dimensions. 
Section (\ref{examples}) contains some examples to illustrate what these conditions mean for simple mechanical systems. 
In particular we discuss why the equation of motion of one dimensional damped oscillator cannot be obtained directly from a Lagrangian and how 
the situation can be remedied by multiplying the equation of motion by an \textit{integrating factor} 
often referred to as Jacobi's last multiplier\cite{leonardo}.  We conclude with a brief discussion in section (\ref{conclusions}). 
Note that we will \textit{not} be using Einstein summation convention. Throughout the paper we will explicitly indicate summations. 
Unless otherwise stated the Latin index ($i,j,k \dots$) shall run from $1$ to $n$.


\section{Helmholtz Conditions} \label{proof}
Suppose a system with $n$ degrees of freedom be described by $n$ second order differential equations.
Let us denote the second order differential equations as 
\begin{equation} \label{forceeqn}
F_i(t,x_j,x'_j,x''_j)=0 
\end{equation}
where prime $(')$ denotes derivative with respect to time $t$.

Helmholtz conditions form a set necessary and sufficient conditions to be satisfied by (\ref{forceeqn}) such that they are given by Euler Lagrange equation 
corresponding to a Lagrangian $L(t,x_j,x'_j)$. These conditions are given as,
\begin{align} 
\label{helm1}
\frac{\partial F_i}{\partial x''_j}=&\frac{\partial F_j}{\partial x''_i} \hspace{12em}\mbox{\bf{H1}}\\
\label{helm2}
\frac{\partial F_i}{\partial x_j}-\frac{\partial F_j}{\partial x_i}=
&\frac{1}{2}\frac{d}{dt}\left[\frac{\partial F_i}{\partial x'_j}-\frac{\partial F_j}{\partial x'_i}\right]
\hspace{6em}\mbox{\bf{H2}} \\
\label{helm3}
\frac{\partial F_i}{\partial x'_j}+\frac{\partial F_j}{\partial x'_i}=& 2\frac{d}{dt}\left[\frac{\partial F_j}{\partial x''_i}\right] 
\hspace{9em}\mbox{\bf{H3}}
\end{align}
A proof for the above statement clearly requires us to show two things, necessity and sufficiency.


\subsection{Necessity}

In order to show the necessity of these conditions, let us consider a Lagrangian $L(t,x_j,x'_j)$ whose Euler-Lagrange's 
equations of motion are denoted by $E_i(t,x_j,x'_j,x''_j)$. Therefore 
\begin{equation} \label{eulerlag}
E_i\equiv\frac{d}{dt}\left(\frac{\partial L}{\partial x'_i}\right) - \frac{\partial L}{\partial x_i}=0 
\end{equation}
where,
\begin{equation}\label{defddt}
\frac{d}{dt}=\frac{\partial }{\partial t} +\sum\limits_{j=1}^{n}\frac{\partial }{\partial x_j}x'_j + \sum\limits_{j=1}^{n}\frac{\partial }{\partial x'_j}x''_j 
\end{equation}
The variations of $E_i$ under arbitrary variation produced along one of the solution curves is 
given as
\begin{equation}\label{variation}
\delta E_i=\sum\limits_{k=1}^{n}\left[ 
\frac{\partial E_i}{\partial x_k}\delta x_k + \frac{\partial E_i}{\partial x'_k}\delta x'_k + \frac{\partial E_i}{\partial x''_k}\delta x''_k  
\right]
\end{equation}
The logic of our proof for the necessity of Helmholtz conditions will be as follows: We shall show that any $E_i$ which
are Euler-Lagrange's equations of motion of some Lagrangian $L$ necessarily satisfy the relations implied by the
Helmholtz conditions ({\bf{H1-H3}}). Therefore if the equations of motion $F_i$ have to be obtainable from some Lagrangian, they also need to
satisfy the same conditions. Note that this is just a necessary condition, not sufficient.

Expanding the equation (\ref{variation}) using (\ref{eulerlag}) and (\ref{defddt}) we get
\begin{align} \label{eulervariation}
\delta E_i=&\sum\limits_{k,j=1}^{n}\left\{\frac{\partial ^3 L}{\partial t \partial x_k \partial x'_i} +
			\frac{\partial ^3 L}{\partial x_k \partial x_j \partial x'_i}x'_j +
			\frac{\partial ^3 L}{\partial x_k \partial x'_j \partial x'_i}x''_j - 
                         \frac{\partial ^2 L}{\partial x_k \partial x_i} \right\}\delta x_k + 						
                                                                                                           \nonumber\\						
		    &\sum\limits_{k,j=1}^{n}\left\{\frac{\partial ^3 L}{\partial t \partial x'_k \partial x'_i} +
			\frac{\partial ^3 L}{\partial x'_k \partial x_j \partial x'_i}x'_j +
			\frac{\partial ^3 L}{\partial x'_k \partial x'_j \partial x'_i}x''_j + 
			\frac{\partial ^2 L}{\partial x_k \partial x'_i}  - 
                       \frac{\partial ^2 L}{\partial x'_k \partial x_i} \right\}\delta x'_k + 			\nonumber\\ 	
		    &\sum\limits_{k,j=1}^{n}\left\{\frac{\partial ^2 L}{\partial x'_k \partial x'_i} \right\}\delta x''_k    
\end{align}
where we have used the fact that the Lagrangian $L$ is a function of the positions and velocities and time only and that
they can all be considered independent of each other. In other words terms like $\frac{\partial x'_j}{\partial x_k} = 0$
and $\frac{\partial x'_j}{\partial x'_k} = \delta_{kj}$ and $\frac{\partial x''_j}{\partial x''_k} =\delta_{kj}$.

The first Helmholtz condition is easy to verify. Consider the coefficient of $\delta x''_k$ in (\ref{eulervariation}). We can see that
\begin{eqnarray}\label{H1necessary}
\frac{\partial E_i}{\partial x''_k} = \frac{\partial ^2 L}{\partial x'_k \partial x'_i} ~ ~ ~ ~ ~
&\mbox{and}& ~ ~ ~ ~ ~ \frac{\partial E_k}{\partial x''_i} = \frac{\partial ^2 L}{\partial x'_i \partial x'_k}
\nonumber\\
\mbox{since partial derivatives commute} ~ ~ ~ ~ ~ \Rightarrow ~ ~ ~  \frac{\partial E_i}{\partial x''_k} &=&
\frac{\partial E_k}{\partial x''_i} 
\end{eqnarray}
which is same as {\bf{H1}} (eq.\ref{helm1}).

To derive {\bf{H2}}, consider (\ref{eulervariation}) again. It can be easily seen that
\begin{equation}\label{H2intermed}
\frac{d}{dt}\left[\frac{\partial E_i}{\partial x'_k}-\frac{\partial E_k}{\partial x'_i}\right]= 
2\frac{d}{dt}\left[\frac{\partial ^2 L}{\partial x_k \partial x'_i}- \frac{\partial ^2 L}{\partial x_i \partial x'_k}\right] 
\end{equation}
since only the last two terms in the coefficient of $\delta x'_k$ do not vanish under antisymmetrization due to interchange of $i\longleftrightarrow k$. Similarly, since only the last term of the coefficient of $\delta x_k$ vanishes
under the same conditions, a simple calculation using the definition (\ref{defddt}) to expand the RHS of (\ref{H2intermed}) gives
\begin{eqnarray}\label{H2necessary}
\frac{d}{dt}\left[\frac{\partial E_i}{\partial x'_k}-\frac{\partial E_k}{\partial x'_i}\right]= 
2\left[\frac{\partial E_i}{\partial x_k}-\frac{\partial E_k}{\partial x_i}\right] 
\end{eqnarray}
The third condition  {\bf{H3}} can be showed by a similar calculation.
\begin{align}\label{H3necessary} 
\frac{d}{dt}\left[\frac{\partial E_i}{\partial x''_k}\right]=&\frac{d}{dt}\left[\frac{\partial ^2 L}{\partial x'_k \partial x'_i} \right] \nonumber \\
=& \frac{1}{2}\left[\frac{\partial E_i}{\partial x'_k} +\frac{\partial E_k}{\partial x'_i}\right]
\end{align}
where now the only the last two terms in the coefficient of $\delta x'_k$ vanish under symmetrization due to
interchange of $i\longleftrightarrow k$.

To summarize, we have proved that the three Helmholtz conditions are relations between the equations of motion obtained
from Euler Lagrange equation of any Lagrangian. It can be seen that these conditions incorporate all possible relations between  
the coefficients in eq (\ref{eulervariation}). Hence if any second order differential equation is an Euler Lagrange's 
equation of some Lagrangian it necessarily needs to satisfy the three Helmholtz conditions. Note that the proof given above works
for all dimensions.

\subsection{Sufficiency}

To prove sufficiency, we need to show that given the equations of motion satisfying the Helmholtz conditions, 
there always exists a Lagrangian for the system. The sufficiency conditions, in some sense are \textit{integrability} conditions 
for the equations of motion. 

Our plan is the following: If $F_i$ are the equations of motion, the Helmholtz conditions puts some restrictions on the form of $F_i$ 
which we determine at first. Then we assume that there exists a Lagrangian $L$ whose Euler-Lagrange's equations of motion are the set $F_i$. 
The restrictions on $F_i$ from the Helmholtz conditions puts restrictions on the form of $L$. That is what we determine next and express 
the hypothetical Lagrangian in terms of some functions we call $G_0, H_i, H_0$. If and only if these functions exist the Lagrangian $L$ will exist. 
We prove the existence of these functions thereby guaranteeing that there exists a Lagrangian $L$ whose Euler-Lagrange's equations of motion are 
given by $F_i$.

As mentioned above, we first show that the Helmholtz conditions imply a specific relations among the set of equations of
motion $F_i$. Let us expand the equation \mbox{\bf(H3)}
\begin{equation} \label{forcehelm} 
\frac{\partial F_i}{\partial x'_j}+\frac{\partial F_j}{\partial x'_i}= 2\left[\frac{\partial^2 F_j}{\partial t\partial x''_i} + \frac{\partial^2 F_j}{\partial x_k\partial x''_i}x'_k\ +  \frac{\partial^2 F_j}{\partial x'_k\partial x''_i}x''_k\ +  \frac{\partial^2 F_j}{\partial x''_k\partial x''_i}x'''_k\right] 
\end{equation}
We observe that left hand side of equation (\ref{forcehelm}) is independent of $x'''_j$ therefore the coefficients of $x'''_j$  on the right hand 
side must identically vanish. Thus, $F_i(t,x_j,x'_j,x''_j)$ must be linear in $x''_j$ and must take the following form
\begin{equation} \label{functionform} 
F_i\equiv P_i(t,x_k,x_k') + \sum_j Q_{ij}(t,x_k,x_k')x''_j
\end{equation}
From \mbox{\bf(H1)} we can clearly see that $Q_{ij}=Q_{ji}$.

Putting $F_i$ in \mbox{\bf(H2)}, we see that the following conditions must be identically satisfied:
\begin{enumerate}[(1)]
\item Coefficients of $x'''_j$ in right hand side of \mbox{\bf(H2)} must vanish. This implies 
\begin{align} 
\label{Qpartialij} 
\frac{\partial Q_{ik}}{\partial x'_j}=\frac{\partial Q_{jk}}{\partial x'_i}
\end{align}
\item The Coefficients of $x''_k$ must be equal
\begin{align} 
\label{coeffequal} 
\frac{\partial Q_{ik}}{\partial x_j}-\frac{\partial Q_{jk}}{\partial x_i}=\frac{1}{2}\left[\frac{\partial P_{i}}{\partial x'_j\partial x'_k}-\frac{\partial P_{j}}{\partial x'_i\partial x'_k}\right]
\end{align}
\item The remaining part of equation is given as
\begin{align} 
\label{Pcoefficient} 
\frac{\partial P_{i}}{\partial x_j}-\frac{\partial P_{j}}{\partial x_i}=\frac{1}{2}\left[\frac{\partial^2 P_{i}}{\partial t\partial x'_j}-\frac{\partial^2 P_{j}}{\partial t\partial x'_i}+\sum\limits_{k=1}^{n}\left\{\frac{\partial^2 P_i}{\partial x_k\partial x'_j}-\frac{\partial^2 P_{j}}{\partial x_k\partial x'_i}\right\}x'_k\right]
\end{align}
\end{enumerate}
On differentiating (\ref{Pcoefficient}) with respect to $x'_l$ we get
\begin{align} 
\label{cycle} 
\frac{\partial^2 P_{i}}{\partial x_j\partial x'_l}-\frac{\partial P^2_{j}}{\partial x_i\partial x'_l}=\frac{1}{2}\left[\frac{\partial^3 P_{i}}{\partial t\partial x'_j\partial x'_l}-\frac{\partial^3 P_{j}}{\partial t\partial x'_i\partial x'_l}+\sum\limits_{k=1}^{n}\left\{\frac{\partial^3 P_i}{\partial x_k\partial x'_j\partial x'_l}-\frac{\partial^3 P_{j}}{\partial x_k\partial x'_i\partial x'_l}\right\}x'_k+\frac{\partial^2 P_i}{\partial x_l\partial x'_j}-\frac{\partial^2 P_{j}}{\partial x_l\partial x'_i}\right]
\end{align}
Interchanging $i\longleftrightarrow j\longleftrightarrow l$ in cyclic order in (\ref{cycle}) and adding the three equations so obtained, we have
\begin{align}\label{cyclecondt}
\frac{\partial^2 P_i}{\partial x_l\partial x'_j}-\frac{\partial^2 P_{j}}{\partial x_l\partial x'_i}=\frac{\partial^2 P_i}{\partial x_j\partial x'_l}-\frac{\partial^2 P_{l}}{\partial x_j\partial x'_i}+\frac{\partial^2 P_l}{\partial x_i\partial x'_j}-\frac{\partial^2 P_{j}}{\partial x_i\partial x'_l}
\end{align}  
Replacing $l$ with $k$ in (\ref{cyclecondt}), (\ref{Pcoefficient}) can now be written as
\begin{align} 
\label{cyclefinal} 
\frac{\partial P_{j}}{\partial x_i}-\frac{\partial P_{i}}{\partial x_j}+\frac{1}{2}\left[\frac{\partial^2
P_{i}}{\partial t\partial x'_j}-\frac{\partial^2 P_{j}}{\partial t\partial
x'_i}+\sum\limits_{k=1}^{n}\left\{\frac{\partial^2 P_i}{\partial x_j\partial x'_k}-\frac{\partial^2 P_{k}}{\partial
x_j\partial x'_i}+\frac{\partial^2 P_k}{\partial x_i\partial x'_j}-\frac{\partial^2 P_{j}}{\partial x_i\partial
x'_k}\right\}x'_k\right]=0
\end{align}
Having established these relations, let us now assume that there exists a Lagrangian function $L(t,x_j,x'_j)$ such that 
the Euler-Lagrange's equations of motion (\ref{eulerlag}) are of the form given by (\ref{functionform}). That is,
\begin{equation} \label{comparison}
\sum_j \left( \frac{\partial ^2 L}{\partial x'_i \partial x'_j}x''_j + \frac{\partial ^2 L}{\partial x'_i \partial
x_j}x'_j + \frac{\partial ^2 L}{\partial x'_i \partial t} - \frac{\partial L}{\partial x_i}\right) = 
\sum_j Q_{ij}(t,x_k,x_k')x''_j + P_i(t,x_k,x_k')
\end{equation}
We need to show that such a Lagrangian exists. That is what we will do subsequently, using the relations derived above. 
Comparing the coefficients of $x''_j$  on both sides of eq. (\ref{comparison}), we have
\begin{eqnarray}
\frac{\partial ^2 L}{\partial x'_i \partial x'_i} &=& Q_{ii} \label{Qii} \\
\frac{\partial ^2 L}{\partial x'_i \partial x'_j} &=& Q_{ij} ~ ~; ~ ~ i \neq j \label{Qij}
\end{eqnarray}
where we have written the two possible types of terms separately. Integrating eqs. (\ref{Qii}) and (\ref{Qij}) we get,
respectively
\begin{eqnarray}
\frac{\partial L}{\partial x'_i} &=& \int\limits_{x'_{o_i}}^{x'_i}{Q_{ii}}{dx'_i} + h_{ii}(t,x_i,x_j,x'_j) 
~ ~ ; ~ ~ i \neq j  \label{partialQii} \\
\frac{\partial L}{\partial x'_i} &=& \int\limits_{x'_{o_j}}^{x'_j}{Q_{ij}}{dx'_j} + 
h_{ij}(t,x_i,x'_i,x_j,x_k,x'_k)  
~ ~ ; ~ ~ i \neq j \neq k \label{partialQij}
\end{eqnarray}
where, the functions $h_{ii},h_{ij}$ are constants of integration while the points $x'_{o_i}$ and $x'_{o_j}$ are arbitrary.
Note that for a $n$ dimensional system we will have $n$ independent equations of the type (\ref{partialQii}) and $n(n-1)$
independent equations of the type (\ref{partialQij}). That gives us $n^2$ functions $h_{ij}$ in the above equations.
 
Partially differentiating eq (\ref{partialQii}) with respect to $x'_j$ we get 
\begin{eqnarray}
\frac{\partial^2 L}{\partial x'_j \partial x'_i } &=& \int\limits_{x'_{o_i}}^{x'_i} {\frac{\partial Q_{ii}}{\partial
x'_j}{dx'_i}} + \frac{\partial h_{ii}}{\partial x'_j} \nonumber \\
&=&  \int\limits_{x'_{o_i}}^{x'_i} {\frac{\partial Q_{ji}}{\partial
x'_i}{dx'_i}} + \frac{\partial h_{ii}}{\partial x'_j} ~ ~ ~  ~ ~ ~ ~\mbox{using eq (\ref{Qpartialij})} \nonumber \\
&=& Q_{ij} + \frac{\partial h_{ii}}{\partial x'_j}
\end{eqnarray}
Comparing this with eq (\ref{Qij}), we see that $\frac{\partial h_{ii}}{\partial x'_j} = 0$. This can be done for all $j$
showing that $h_{ii}$ is independent of all the velocities.

Similarly partially differentiating eq (\ref{partialQij}) with respect to $x'_i$ and again using eq (\ref{Qpartialij}) 
and comparing with eq (\ref{Qii}), we can show that  $\frac{\partial h_{ij}}{\partial x'_i} = 0$. However $h_{ij}$ can still
be a function of $x'_k$ where $k \neq i \neq j$. Then eq (\ref{partialQij}) becomes
\begin{eqnarray}
\frac{\partial L}{\partial x'_i} &=& \int\limits_{x'_{o_j}}^{x'_j}{Q_{ij}}{dx'_j} + h_{ij} (t,x_i,x_j,x_k,x'_k)  
~ ~ ; ~ ~ i\neq j\neq k
\label{partialQij2ndpass}
\end{eqnarray}
However partially differentiating with respect to $x'_k$ and this time using eq (\ref{Qpartialij}) and repeating the above
procedure, we can see that $\frac{\partial h_{ij}}{\partial x'_k} = 0$. Therefore $h_{ij}$ is also independent of all the
velocities.

Next, note that for each $\frac{\partial L}{ \partial x'_i}$, we have $n$ equations corresponding to each $j$ taking values from $1$ to $n$.
\begin{eqnarray} 
&& \frac{\partial L}{ \partial x'_i} =
\int\limits_{x'_{o_j}}^{x'_j}{\frac{\partial^2 L}{\partial x'_j \partial x'_i }} dx'_j + h_{ij}(t,x)  ~ ~ ~;~ ~
~\mbox{where, x stands for all generalized coordinates} \label{delLdelxdoti}
\end{eqnarray}
Adding all the equations for each $ \frac{\partial L}{ \partial x'_i}$ and dividing by $n$,
\begin{eqnarray}\label{intarray1}
\frac{\partial L}{ \partial x'_i} &&=
\sum\limits_{j=1}^{n}\frac{1}{n}\left[ \int\limits_{x'_{o_j}}^{x'_j}{\frac{\partial^2 L}{\partial x'_j \partial x'_i }} dx'_j + h_{ij}(t,x)\right]\nonumber \\
&&=\sum\limits_{j=1}^{n}\frac{1}{n}\left[\int\limits_{x'_{o_j}}^{x'_j}{Q_{ij}} dx'_j \right] +
G_i(t,x)  ~~~~~;~~~~\mbox{where, $G_i(t,x) = \frac{1}{n}\sum\limits_{j=1}^{n}h_{ij}(t,x)$}\nonumber \\
&&=\frac{1}{n}\left[\int\limits_{x'_{o_1}}^{x'_1}\int\limits_{x'_{o_2}}^{x'_2}\dots\int\limits_{x'_{o_n}}^{x'_n}\sum\limits_{j=1}^{n}{{Q_{ij} dx'_j}}\right]
+ G_i(t,x) 
\end{eqnarray}
From eq (\ref{Qpartialij}) we have
$\frac{\partial Q_{ik}}{\partial x'_j}=\frac{\partial Q_{ij}}{\partial x'_k}$. Hence $\sum\limits_{j=1}^{n}Q_{ij}dx'_j$ 
form an exact differential whose integration is path independent and is guaranteed to exist 
(see Appendix {(\ref{appendixa})}). To simplify our notation we will subsequently denote 
\begin{eqnarray} 
R_i(t,x,x') := 
\frac{1}{n}\left[\int\limits_{x'_{o_1}}^{x'_1}\int\limits_{x'_{o_2}}^{x'_2}\dots\int\limits_{x'_{o_n}}^{x'_n}\sum\limits_{j=1}^{n}{{Q_{ij} dx'_j}}\right]
\label{Rdefn}
\end{eqnarray}

Then the n equations in eq (\ref{delLdelxdoti}) can be rewritten as 
\begin{equation} \label{redintarray1}
\frac{\partial L}{\partial x'_i}=R_i(t,x,x') + G_i(t,x) 
\end{equation}

We can follow a similar path to perform one more integral. Again adding up the $n$ equations in (\ref{redintarray1}) and
diving by $n$, we get
\begin{eqnarray}\label{intarray2}
 L =&& \frac{1}{n}\left[\sum\limits_{i=1}^{n}\int\limits_{x'_{o_i}}^{x'_i}{\frac{\partial L}{ \partial x'_i}}dx'_i\right] \nonumber\\
 =&& \frac{1}{n}\left[\int\limits_{x'_{o_1}}^{x'_1}\int\limits_{x'_{o_2}}^{x'_2}\dots\int\limits_{x'_{o_n}}^{x'_n}\sum\limits_{i=1}^{n}\frac{\partial L}{ \partial x'_i}dx'_i\right] \nonumber\\
 =&& \frac{1}{n} \left[\int\limits_{x'_{o_1}}^{x'_1}\int\limits_{x'_{o_2}}^{x'_2}\dots\int\limits_{x'_{o_n}}^{x'_n}\sum\limits_{i=1}^{n}R_i(t,x,x')dx'_i\right]+\frac{1}{n}\sum\limits_{i=1}^{n}G_i(t,x)x'_i + H_0(t,x) 
\end{eqnarray}
where, $H_0(t,x)$ is a constant of integration. 

Using  eqs (\ref{Qpartialij}) we can verify that
\begin{align} \label{exactness}
\frac{\partial R_i}{\partial x'_j}= \frac{\partial R_j}{\partial
x'_i}
\end{align}
thus proving that $\sum\limits_{i=1}^{n}R_i(t,x,x')dx'_i$ is an exact differential whose solution exists and is path
independent (see Appendix {(\ref{appendixa})}). 

Again, for simplicity, we shall use the notation
\begin{eqnarray}
G_0(t,x,x') &:=&
\frac{1}{n}
\left[\int\limits_{x'_{o_1}}^{x'_1}\int\limits_{x'_{o_2}}^{x'_2}\dots\int\limits_{x'_{o_n}}^{x'_n}\sum\limits_{i=1}^{n}R_i(t,x,x')dx'_i\right]
\nonumber \\
H_i(t,x') &:=& \frac{1}{n}G_i(t,x') \label{GHdefn} 
\end{eqnarray}

The we can write our hypothetical Lagrangian as
\begin{eqnarray}
L = G_0(t,x,x') + \sum\limits_{i=1}^{n} H_i (t,x)x'_i + H_0(t,x)
\label{L2}
\end{eqnarray}
Note that, what we have proved above is that  $G_0(t,x,x')$ will exist provided the Helmholtz conditions are satisfied. 
So we have reduced our problem of proving the existence of Lagrangian to proving the existence of functions $H_i$ and $H_0$.

We notice that these functions must satisfy the remaining part of equation (\ref{comparison}) that is,
\begin{align} 
\sum\limits_{j=1}^{n}\frac{\partial ^2 L}{\partial x'_i \partial x_j}x'_j + \frac{\partial ^2 L}{\partial x'_i \partial t} - \frac{\partial L}{\partial x_i}=& P_i(t,x_k,x_k') \\
\intertext{which gives us n differential equations,}
\label{partial}
\frac{\partial H_i}{\partial t}+\sum\limits_{j=1}^{n}\left[\frac{\partial H_i}{\partial x_j}-\frac{\partial H_j}{\partial x_i}\right]x'_j - \frac{\partial H_0}{\partial x_i} 
=& 
P_i + \frac{\partial G_0}{\partial x_i} - \left[\sum\limits_{j=1}^{n}\frac{\partial ^2 G_0}{\partial x_j \partial x'_i}x'_j\right] -\frac{\partial ^2 G_0}{\partial t \partial x'_i}
\end{align}
Note that, from the equations of motion (\ref{functionform}) we already know the functions $P_i$ and $Q_{ij}$. As seen
above, given the functions $Q_{ij}$, we can always find the functions $G_0$. Hence the RHS of the above equation is given
completely in terms of known functions. Since the LHS of the above equation (\ref{partial}) must be identically equal to its RHS, 
we can verify that the RHS is independent of $x'_i$ and linearly dependent on $x'_j$ for all $j \neq i$ simply from the
definitions (\ref{GHdefn}). 

Partially differentiating eq (\ref{partial}) with respect to $x'_k$ we get
\begin{align} 
\label{a}
\frac{\partial H_i}{\partial x_k}-\frac{\partial H_k}{\partial x_i}= \frac{\partial P_i}{\partial x'_k} + \frac{\partial^2 G_0}{\partial x'_k\partial x_i} - \left[\sum\limits_{j=1}^{n}\frac{\partial ^3 G_0}{\partial x_j \partial x'_i \partial x'_k}x'_j\right]-\frac{\partial ^2 G_0}{\partial x_k \partial x'_i} -\frac{\partial ^3 G_0}{\partial t \partial x'_i \partial x'_k}
\end{align}  
Interchanging $i\longleftrightarrow k$ in (\ref{a}) and subtracting, we have
\begin{align} 
\label{b}
\frac{\partial H_i}{\partial x_k}-\frac{\partial H_k}{\partial x_i}= \frac{1}{2}\left[\frac{\partial P_i}{\partial x'_k}-\frac{\partial P_k}{\partial x'_i}\right] + \frac{\partial^2 G_0}{\partial x'_k\partial x_i} -\frac{\partial ^2 G_0}{\partial x_k \partial x'_i} \end{align}  
giving us a first order partial differential equation for pairs of $H_i$.

On substituting (\ref{b}) in (\ref{partial}), we have,
\begin{align}
\label{c}
\frac{\partial H_i}{\partial t}- \frac{\partial H_0}{\partial x_i} =& P_i + \frac{\partial G_0}{\partial x_i} - \sum\limits_{j=1}^{n}\left[\frac{\partial ^2 G_0}{\partial x_i \partial x'_j}x'_j\right] -\frac{\partial ^2 G_0}{\partial t \partial x'_i} + \sum\limits_{j=1}^{n}\frac{1}{2}\left[\frac{\partial P_j}{\partial x'_i} -\frac{\partial P_i}{\partial x'_j}\right]x'_j
\end{align}  
giving us a first order partial differential equation involving $H_0$ and $H_i$.

We have $\frac{n(n-1)}{2}$ equations of type (\ref{b}) and $n$ equations of type (\ref{c}), giving us a total of
$\frac{n(n+1)}{2}$ equations. Since the RHS of (\ref{b}) and (\ref{c}) are known, let us denote them as $\phi_{ik}$ and
$\theta_i$ respectively. It can be easily seen that for each pair $i,k$ 
\begin{align}
\frac{\partial \phi_{ik}}{\partial t}-\frac{\partial \theta_i}{\partial x_k}+\frac{\partial \theta_k}{\partial x_i}\equiv0
\end{align}
since partial derivatives involving $G_0$ cancel each other and the remaining part is zero by (\ref{cyclefinal}). As
shown by the theorem given in Appendix (\ref{appendixb}), this system of partial differential equations is guaranteed to
have a solution. Thus we have proved that the function $H_0$ and $H_i$ exist and therefore our hypothetical Lagrangian
actually exists. We have finally shown that,  given set of equations of motion, the Helmholtz conditions are sufficient
for the existence of a Lagrangian.  

We further observe that Euler Lagrange equation remains invariant if we add a total derivative of a function 
$f(t,x)$. Hence, for $n$ dimensional system (\ref{functionform}), the most general form of Lagrangian that exists is given as
\begin{equation} 
L = G_0(t,x,x') + \sum\limits_{i=1}^{n} H_i (t,x)x'_i + H_0(t,x) + \frac{d f(t,x)}{dt} 
\end{equation}

\section{1-D Case} \label{examples}

In this section we will apply the mathematical arguments developed hitherto to one dimensional classical systems. 
We notice that for 1-D systems, when the index $i$ takes only one value, the conditions {\bf H1} and {\bf H2} 
(\ref{helm1} and \ref{helm2}) are identically satisfied. Thus, it is sufficient to check only the third  Helmholtz 
condition {\bf H3} (\ref{helm3}). 

As an example let us first consider a simple harmonic oscillator. The equation of motion for simple harmonic oscillator, 
is given by
\begin{equation}\label{SHM}
F(t,x,x',x'')\equiv x''+\omega^2 x=0
\end{equation} 
We observe that this equation of motion satisfies all Helmholtz conditions (\ref{helm1}-\ref{helm3}). This guarantees the existence of 
Lagrangian which we know is given as
\begin{equation}
L(t,x,x')=\left[\frac{x'{}^2}{2} -\frac{\omega^2 x^2}{2}\right] 
\end{equation}

A more interesting case to study is a damped harmonic oscillator whose equation of motion is given by
\begin{equation}\label{DHM}
F(t,x,x',x'')\equiv x''+bx'+\omega^2 x=0
\end{equation}
We notice that, since $\frac{\partial F}{\partial x'} = b$ while 
$\frac{d}{d t}\left[ \frac{\partial F}{\partial x''}\right] = 0$, the equation of motion given in (\ref{DHM}) fails to satisfy 
the third Helmholtz condition (\ref{helm3}). Hence it cannot be obtained from any Lagrangian. 

Lagrangian formulation of damped harmonic oscillator was first studied by \cite{bateman}. 
One of the primary motivations for the continued interest in such systems  (see \cite{dekker} and references therein
as well as more recently \cite{baldoitti}) is that these systems provide toy models to study quantum mechanics of dissipative
systems. The usual routes to quantization start from a Lagrangian or Hamiltonian. Hence if this system does not
admit a Lagrangian formulation, the already difficult problem of dissipative quantum systems becomes even more
difficult. 

So to obtain suitable candidate for Lagrangian of the system \cite{leonardo} 
we modify the equation of motion with a multiplicative factor. Explicitly, instead of (\ref{DHM}), let us start
with an equation of motion given by
\begin{equation}\label{mDHM}
F(t,x,x',x'')\equiv \Lambda(t,x,x')(x''+ bx' + \omega^2 x) \equiv  \Lambda(t,x,x')(x'' + G(t,x,x'))=0
\end{equation}
where $\Lambda(t,x,x')$ is  known as Jacobi last multiplier. Comparing with eq (\ref{functionform}) we see that 
$Q_{ij} \Rightarrow \Lambda$ and $P_i \Rightarrow \Lambda G$ 

Now the the third  Helmholtz condition {\bf H3} (\ref{helm3}) becomes
\begin{align}
\frac{\partial \Lambda}{\partial x'}\left( x'' + G \right) + \Lambda\frac{\partial G}{\partial x'} =&\frac{d\Lambda}{dt}
\label{H3dho}
\end{align}
Since, $x''=-G$ from (\ref{mDHM}) hence, (\ref{H3dho}) can be rewritten as,
\begin{align}
&\frac{d \Lambda}{d t} = \Lambda\frac{\partial G}{\partial x'} \nonumber
\intertext{giving,}
&\Lambda(t,x,x') = \exp{\int\frac {\partial G(t,x,x')}{\partial x'}\mathrm{d}t}  \nonumber
\intertext{Which for the case of damped harmonic oscillator is given as}
&\Lambda(t,x,x')=e^{bt}
\end{align}
Note that this $\Lambda$ satisfies the third Helmholtz condition (\ref{H3dho}).

It can be easily verified that one possible Lagrangian whose Euler Lagrange's equation results in with (\ref{mDHM}) is given as
\begin{equation}\label{LDHO}
L(t,x,x')=e^{bt}\left[\frac{x'{}^2}{2} -\frac{\omega^2 x^2}{2}\right]
\end{equation}
We observe that in the limit $b \rightarrow 0$ i.e. in the case of vanishing damping, this goes over to the standard
Lagrangian for simple harmonic oscillators. This is a standard Lagrangian for describing damped harmonic oscillators
\cite{dekker} although it should be noted that this Lagrangian is not unique. 

\section{Conclusion} \label{conclusions}

In this article, we have presented a pedagogical proof for the inverse problem of Lagrangian dynamics. We have
shown the necessity and sufficiency of Helmholtz conditions for a set of $n$ dimensional second order differential equations 
to be given by Euler Lagrange equations for some Lagrangian function. The proof is somewhat long but is simple and is
accessible to undergraduate students and we feel that it will provide a good supplementary material for undergraduate
Classical Mechanics curriculum.

We also give an illustrative example of a damped harmonic oscillator. Dissipative systems are deceptively easy to
describe and notoriously difficult to formulate, both classically and quantum mechanically and are of interest in
current research. Our example gives some insight into the problems one faces while studying such systems.

\appendix

\section{Exact Differential}\label{appendixa}

This is a standard theorem for partial differential equations (see for eg \cite{PDEbook}) which we give here for
completeness
\begin{theorem}
\label{theorem a.1} Given a differential of the form $f(x_1,x_2)dx_1+g(x_1,x_2)dx_2$ with $f(x_1,x_2)$ 
and $g(x_1,x_2)$ continuously differentiable with continuous first partial derivatives on a simply connected 
open subset D of $\mathbb{R}^2$ then a potential function F such that $dF=f(x_1,x_2)dx_1+g(x_1,x_2)dx_2$ exists iff 
\begin{align}\label{exact}
\frac{\partial f}{\partial x_2}=\frac{\partial g}{\partial x_1}
\end{align}
\end{theorem}
\begin{proof}
In order to proof the necessity of (\ref{exact}) let us evaluate the differential $dF$ as 
\begin{align}\label{potential}
dF(x_1,x_2)=\frac{\partial F}{\partial x_1}dx_1 + \frac{\partial F}{\partial x_2}dx_2
\end{align}
i.e., $f(x_1,x_2)=\frac{\partial F}{\partial x_1}$ and $g(x_1,x_2)=\frac{\partial F}{\partial x_2}$. 
Since, the second partial derivatives of $F(x_1,x_2)$ commute due to continuity of first partial derivatives of $f$ and $g$, we have 
\begin{align}\label{commute}
\frac{\partial^2 F}{\partial x_1\partial x_2}= \frac{\partial^2 F}{\partial x_2\partial x_1}
\end{align}
Hence, proving the necessity of (\ref{exact}).

In order to prove the sufficiency of (\ref{exact}), we integrate  $\frac{\partial F}{\partial x_1}=f(x_1,x_2)$ with respect to $x_1$ to get
\begin{align}\label{findpot1}
F(x_1,x_2)=\int f(x_1,x_2)dx_1  + h(x_2) 
\end{align}
where, $h(x_2)$ is constant of integration. 
To obtain $h(x_2)$ we differentiate (\ref{findpot1}) partially with respect to $x_2$, use $\frac{\partial F}{\partial x_2}=g(x_1,x_2)$ 
and finally integrate back with respect to $x_2$. This gives
\begin{align}\label{findpot2}
h(x_2)=\int g(x_1,x_2)dx_2-\int f(x_1,x_2)dx_1 + k
\end{align}
where $k$ is a numerical constant.
We have completely determined the potential function $F(x_1,x_2)$ thus proving the sufficiency. \\
\end{proof}

This theorem can be easily extended to $n$ dimensions. Consider a differential given by
$\sum\limits_{i=1}^{n}f_i(x)dx_i$ where $x$ now stands for n variables. Suppose that the set of functions $f_i(x)$ 
are continuously differentiable on a simply connected open subset D of $\mathbb{R}^n$ and satisfy the following condition 
identically for all $i,j$
\begin{align}\label{mixed}
\frac{\partial f_i}{\partial x_j}=\frac{\partial f_j}{\partial x_i}
\end{align}
Considering a particular pair $i=p,j=q$, we observe that the expression $f_p(x)dx_p + f_q(x)dx_q$ is exact due to (\ref{mixed}) and hence the potential function $F_{pq}$ exists for it. Now consider, the pairs $i=p,j=r$ and $i=r,j=q$, by the same argument there exist potential functions $F_{pr}$ and $F_{qr}$ respectively. Hence, we have
\begin{align}\label{potential2d}
dF_{pq}=&f_p(x)dx_p + f_q(x)dx_q \nonumber\\
dF_{pr}=&f_p(x)dx_p + f_r(x)dx_r \nonumber\\
dF_{qr}=&f_q(x)dx_q + f_r(x)dx_r 
\end{align}
Adding the last three equation (\ref{potential2d}) and taking $\frac{1}{2}\left[dF_{pq}+dF_{pr}+dF_{qr}\right]=dF_{pqr}$ we have 
\begin{align}\label{potential3d}
dF_{pqr}(x)=f_p(x)dx_p + f_q(x)dx_q + f_r(x)dx_r\nonumber\\
\end{align}
Thus, we can generalize this procedure to n dimensions such that
\begin{eqnarray}\label{potential3dn}
dG(x)=\sum\limits_{i=1}^{n}f_i(x)dx_i  \hspace{3em} \mbox{where} ~ ~ ~ ~
dG(x)=\frac{1}{2}\sum\limits_{\substack{i,j=1\\i < j}}^{n}dF_{ij}
\end{eqnarray}

\section{Existence Theorem}\label{appendixb}

\begin{theorem} \label{theorem b.1} 
Consider a system of differential equations in $\textbf{2+1}$ dimensions $(x_1,x_2,t)$ that satisfies the relations
\begin{align}\label{exist}
\frac{\partial H_1}{\partial x_2} - \frac{\partial H_2}{\partial x_1}=&\phi_{12} \nonumber\\
\frac{\partial H_1}{\partial t} - \frac{\partial H_0}{\partial x_1}=&\theta_1  \nonumber\\
\frac{\partial H_2}{\partial t} - \frac{\partial H_0}{\partial x_2}=&\theta_{2}
\end{align}
where, $H_1,H_2,H_0,\phi_{12},\theta_1,\theta_2$ are all functions of $(x_1,x_2,t)$ and $\phi_{12}=-\phi_{21}$.
The necessary and sufficient condition for the existence of $\textbf{2+1}$ functions $H_1,H_2,H_0$ satisfying the above equation is that the right hand side of (\ref{exist}) must satisfy the following relation:
\begin{align}\label{existcondt}
\frac{\partial \phi_{12}}{\partial t} - \frac{\partial \theta_1}{\partial x_2}+\frac{\partial \theta_2}{\partial x_1}=0
\end{align}
\end{theorem}
\begin{proof}
The necessity of condition (\ref{existcondt}) is evident as it is identically satisfied by LHS of (\ref{exist}). 
In order to proof the sufficiency of (\ref{exact}), let us make the following substitution:
\begin{align}\label{substitute}
A_x=&H_1 ~~ A_y=H_2 ~~ A_z=H_0 \nonumber\\
B_x=&-\theta_2 ~~ B_y=\theta_1 ~~ B_z=-\phi_{12} \nonumber\\ 
\vec{\nabla}=&\frac{\partial}{\partial x_1} + \frac{\partial}{\partial x_2} + \frac{\partial}{\partial t}
\end{align}
On making above substitution, we observe that equation (\ref{exist}) reduces to $\vec{\nabla} \times \vec{A}=\vec{B}$ and 
equation (\ref{existcondt}) reduces to $\vec{\nabla}.\vec{B}=0$. Since, under these conditions $\vec{A}$ always has a
solution which is undetermined upto gradient of a function $f(x_1,x_2,t)$, we can solve for $\vec{A}$ by imposing condition 
analogous to Coulomb gauge of electrodynamics  $\vec{\nabla}.\vec{A}=0$ \cite{griffith}
\begin{eqnarray}
\vec{\nabla} \times (\vec{\nabla} \times \vec{A}) &=&  \vec{\nabla} \times \vec{B} \nonumber\\
\implies  ~ ~ ~ \vec{\nabla}^2\vec{A} &=&-\vec{\nabla} \times \vec{B} \label{solve}
\end{eqnarray} 
The last equation of (\ref{solve}) gives us three Poisson's equation for each of the components of $\vec{A}$ whose
solution, in appropriate units, is given as:
\begin{align}
\vec{A}(\vec{r})=\frac{1}{4\pi}\int\limits_{\vec{r}_0}^{\vec{r}}\frac{\vec{\nabla}' \times \vec{B} (\vec{r'})}{|\vec{r}-\vec{r}'|}d^3r'
\end{align}
where, $\vec{r}=(x_1,x_2,t)$
Hence, we have showed that system (\ref{exist}) always has a solution. \\
\end{proof}

We will now extend this proof for n-dimensional case as showed in \cite{davis2}.
\begin{theorem} \label{theorem b.2} 
Consider a system of differential equations in $\textbf{n+1}$ dimensions $(x_i,t)$ that satisfies the relations
\begin{align}\label{existn}
\frac{\partial H_i}{\partial x_j} - \frac{\partial H_j}{\partial x_i}=&\phi_{ij} \nonumber\\
\frac{\partial H_i}{\partial t} - \frac{\partial H_0}{\partial x_i}=&\theta_i
\end{align}
where, $H_i,H_0,\phi_{ij},\theta_i$ are all functions of $(x_i,t)$, $i,j$ running from $1$ to $n$ and $\phi_{ij}=-\phi_{ji}$.
The necessary and sufficient condition for the existence of $\textbf{n+1}$ functions $H_i,H_0$ satisfying the above equation 
is that the right hand side of (\ref{existn}) must satisfy the following relation:
\begin{align}\label{existcondtn}
\frac{\partial \phi_{ij}}{\partial t} - \frac{\partial \theta_i}{\partial x_j}+\frac{\partial \theta_j}{\partial x_i}=0
\end{align}
for all pairs of $i,j$. 
\end{theorem}
\begin{proof} Let us begin by considering specific values for $i,j$ say $i=1, j=2$ and solve for the resulting system using 
(\ref{theorem b.1}). Say after solving we obtain functions $H_1,H_2,H_0$. Now consider $i=1,j\neq1$ and solve for $H_j$ by 
integrating second equation of (\ref{existn}).
\begin{align}\label{soln}
H_j(t,x)=\int\left[\theta_j(t',x)+\frac{\partial H_0}{\partial x_j}(t',x)\right]dt'
\end{align}
where, as before, $x$ is a symbol for all $x_i's$ and the constant function of integration is chosen to be $0$. 
Let us now check the consistency of our solution i.e we need to check whether the $H_j$ we obtain above satisfies the system (\ref{existn}).

On differentiating (\ref{soln}) partially with respect to $x_1$, we get
\begin{align}\label{consistencycheck1}
\frac{\partial H_j}{\partial x_1}(t,x)=&\int\left[\frac{\partial \theta_j}{\partial x_1}(t',x)+\frac{\partial^2 H_0}{\partial x_j\partial x_1}(t',x)\right]dt' 
\nonumber\\
=&\int\left[\frac{\partial \theta_j}{\partial x_1}+\frac{\partial}{\partial x_j}\left\{\frac{\partial H_1}
{\partial t'}-\theta_1\right\}\right]dt' \nonumber\\
=&\int\left[\frac{\partial \phi_{j1}}{\partial t'}+\frac{\partial^2 H_1}{\partial t'\partial x_j}\right]dt' \nonumber\\
=&\phi_{j1} + \frac{\partial H_1}{\partial x_j}
\end{align}
where, we have used (\ref{existn}) and (\ref{existcondtn}) in second and third step respectively.

Hence, we see that the solution (\ref{soln}) consistently satisfies the system of equations (\ref{existn}) for $i=1,j$. Similarly, we
can verify that the solution is consistent for $i=2,j$. Since, j is arbitrary we conclude that the entire system can be
solved consistently.
\end{proof}

\end{document}